\documentclass[conference]{IEEEtran}

\usepackage{main}
\usepackage{pgf}
\usepackage{import}\usepackage[absolute,showboxes]{textpos}
  
\TPMargin{5pt}
  
\newcommand{\copyrightstatement}{
    \begin{textblock}{0.84}(0.08,0.01)    
         \noindent
         \footnotesize
         \copyright 2023 IEEE. Personal use of this material is permitted. Permission from IEEE must be obtained for all other uses, in any current or future media, including reprinting/republishing this material for advertising or promotional purposes, creating new collective works, for resale or redistribution to servers or lists, or reuse of any copyrighted component of this work in other works.
    \end{textblock}
}
\newcommand{\tf}{d}

\newcommand{\nuser}{n}
\newcommand{\ecri}{L}

\newcommand{\E}{\mathrm{E}}

\newcommand{\gk}[1]{\left\{#1\right\}}
\newcommand{\hk}[1]{^{(#1)}}

\newcommand{\ek}[1]{\left[#1\right]}
\newcommand{\rk}[1]{\left(#1\right)}
\newcommand{\Pfrak}{\mathfrak{P}}
\renewcommand{\l}{\lambda}
\newcommand{\ex}{\mathrm{e}}

\renewcommand{\d}{\mathrm{d}}

\newcommand{\C}{\mathbb{C}}
\newcommand{\N}{\mathbb{N}}
\newcommand{\R}{\mathbb{R}}

\newcommand{\Mcal}   {{\mathcal M }}

\begin{acronym}
\acro{AP}{access point}
\acro{RA}{Random access}
\acro{LT}{Luby Transform}
\acro{BP}{belief propagation}
\acro{i.i.d.}{independent and identically distributed}
\acro{IoT}{Internet of Things}
\acro{mMTC}{massive machine type communication}
\acro{SA}{Slotted ALOHA}
\acro{PER}{packet error rate}
\acro{SIC}{successive interference cancellation}
\acro{URLLC}{ultra-reliable low latency communication}
\acro{PMF}{probability mass function}
\acro{CRI}{collision resolution interval}
\acro{PGF}{probability generating function}
\acro{CPGF}{conditional probability generating function}
\acro{EGF}{exponential generating function}
\acro{TGF}{transformed generating function}
\acro{rhs}{right-hand side}
\acro{lhs}{left-hand side}
\acro{CSA}{coded slotted ALOHA}
\acro{MPR}{multi-packet reception}
\acro{MTA}{modified binary-tree algorithm}
\acro{SICTA}{binary tree-algorithm with successive interference cancellation}
\acro{BTA}{binary tree-algorithm}
\acro{MST}{maximum stable throughput}
\acro{MTA}{modified tree-algorithm}
\acro{CRP}{collision-resolution protocol}
\acro{CAP}{channel-access protocol}
\acro{IRSA}{Irregular Repetition Slotted ALOHA}
\acro{CSA}{Coded Slotted ALOHA}
\acro{STA}{Simple tree algorithm}
\acro{IC}{Interference Cancellation}
\acro{ATIC}{Advanced Tree-algorithm with Interference Cancellation}
\acro{r.v.}{random variable}
\end{acronym}

\title{An Advanced Tree Algorithm with Interference Cancellation in Uplink and Downlink}

\author{\IEEEauthorblockN{Quirin Vogel\IEEEauthorrefmark{1},Yash Deshpande\IEEEauthorrefmark{2}, \v Cedomir Stefanovi\' c\IEEEauthorrefmark{3}, Wolfgang Kellerer\IEEEauthorrefmark{2}}\\
\IEEEauthorblockA{
\IEEEauthorrefmark{1}Department of Mathematics, School of Computation, Information and Technology, Technical University of Munich, Germany\\
\IEEEauthorrefmark{2}Chair of Communication Networks, School of Computation, Information and Technology, Technical University of Munich, Germany\\
\IEEEauthorrefmark{3}Department of Electronic Systems, Aalborg University, Denmark \\
Email:\{quirin.vogel, yash.deshpande,wolfgang.kellerer\}@tum.de, cs@es.aau.dk}}

\begin{document}
\copyrightstatement
\maketitle

\begin{abstract}

In this paper, we propose \ac{ATIC}, a variant of \ac{SICTA} introduced by Yu and Giannakis.
\ac{ATIC} assumes that \ac{IC} can be performed both by the \ac{AP}, as in \ac{SICTA}, but also by the users. 
Specifically, after every collision slot, the \ac{AP} broadcasts the observed collision as feedback. Users who participated in the collision then attempt to perform \ac{IC} by subtracting their transmissions from the collision signal.
This way, the users can resolve collisions of degree 2 and, using a simple distributed arbitration algorithm based on user IDs, ensure that the next slot will contain just a single transmission.
We show that \ac{ATIC} reaches the asymptotic throughput of 0.924 as the number of initially collided users tends to infinity and reduces the number of collisions and packet delay. 
We also compare \ac{ATIC} with other tree algorithms and indicate the extra feedback resources it requires. 

\end{abstract}


\acresetall
\section{Introduction}
\label{sec:introduction}
\ac{RA} algorithms are crucial for managing a large number of devices that make up the \ac{IoT} ecosystem.
IoT scenarios are often characterized by sparse, intermittent packet arrivals and short packet duration~\cite{5gtrafficsurvey}. 
In contrast to scheduling approaches, \ac{IoT} devices via \ac{RA} algorithms can transmit their data packets without having to wait for a predetermined time slot, eliminating the signalling overhead required to register devices and schedule packets.
In particular, by efficiently managing the transmission of short packets, \ac{RA} algorithms can help reduce the overall latency of the \ac{IoT} network, improving the user experience for applications such as real-time monitoring and control~\cite{Vilgelm2021}.
The main challenge for \ac{RA} algorithms is to avoid collisions that occur when multiple devices try to transmit their data simultaneously, and the major families of \ac{RA} algorithms differ in the way how the collisions are handled. 


The idea behind tree algorithms~\cite{capetanakis1979tree, massey1981collision}  is to handle collisions by successively partitioning the colliding users into smaller groups. 
Execution of tree algorithms requires broadcast of the feedback from the \ac{AP} after every uplink slot.
They exhibit stability until a certain value of the aggregated user arrival rate per uplink slot, which is denoted as the \ac{MST}.
For the basic variant of the algorithm, which is \ac{BTA}, the \ac{MST} is 0.346 packets per slot ~\cite{capetanakis1979tree}.
Through optimizations of the splitting factor and the channel access policy, the \ac{MST} of tree algorithms can be increased to 0.4878 packets/slot ~\cite{verdu1986}.

Advanced signal processing techniques, such as interference cancellation, further improve the \ac{MST} of tree-algorithms~\cite{yu2005sicta}. 
Specifically, in~\ac{SICTA}, 
after decoding a user packet, the \ac{AP} successively applies interference cancellation over the previously received and stored collision slots to remove the interference contribution of the decoded user, which potentially enables the decoding of new packets and new rounds of \ac{SIC}.
The \ac{MST} of \ac{SICTA} is 0.693 packets/slot.

In this paper, we propose an algorithm called \ac{ATIC}, 
    that, leveraging enhanced feedback and \ac{IC} capabilities at the users, resolves collisions of degree 2 using just 1 extra uplink slot.
    This way, the algorithm can achieve \ac{MST} of 0.924, as shown in the paper, 
     exceeding the \ac{MST} of \ac{SICTA} by over 33\%. 
    We also show that gated access is the best access method for the proposed algorithm.
    To compare the extra feedback cost of the proposed algorithm, we assess the amount of resources required in the feedback for the relevant classes of tree algorithms. 
    Finally, we show that the mean packet delay\footnote{The mean number of uplink slots needed to successfully decode the packet after its arrival at a user.} of the proposed algorithm is significantly less than the one of SICTA for significantly higher arrival rates. 

The rest of the paper is organized as follows.
Section~\ref{sec:backgound} provides the background and an overview of the related work.
Section~\ref{sec:system_model} states the system model explains the proposed algorithm \ac{ATIC}. 
Section~\ref{sec:analysis} derives the \ac{MST} of \ac{ATIC}.
We show some simulation-based evaluation of \ac{ATIC} in terms of the required resources and mean packet delay in Section~\ref{sec:evaluation}. 
Finally, we conclude the paper by discussing our findings in Section~\ref{sec:conclusion}. 

\section{Background and Motivation}
\label{sec:backgound}



\subsection{Standard Tree Algorithms}

The \ac{BTA}, also known as the Capetanakis-Tsybakov-Mikhailov type \ac{CRP} was proposed in~\cite{capetanakis1979tree}. 
In \ac{BTA}, the users decide on which slot to transmit based on the feedback from the \ac{AP}.
The feedback 
for any slot $t$ is ternary $F[t] \in \{0,1,e\},\,\, \forall t \in \mathbb{Z}^{+}$, denoting whether the outcome of the slot was an idle $F[t]=0$, a success $F[t]=1$, or a collision $F[t]=e$, respectively. 

The \ac{CRP} starts by $\nuser$ users transmitting their packets for the first time in the coming slot, thus marking the beginning of a \ac{CRI} (for which the slot counter $t$ is set to 0).
The moments of the initial transmissions are determined by the \ac{CAP}, which will be elaborated later.
If $\nuser \geq 2$, the slot is a collision, i.e., $F[t]=e$, and the $\nuser$ users independently split into two groups, e.g. group 0 and group 1, according to the outcome of a Bernoulli trial.
The users who have joined group 0 transmit in the next slot, $t+1$, while users who have joined group 1 abstain from transmitting and observe the feedback until users from group 0 are resolved.
If this slot is also a collision, $F[t+1]=e$, then the process of partitioning the users further is done recursively.
Users who have joined group 1 wait until all the users in group 0 have successfully transmitted their packets to the \ac{AP}, after which they transmit in the next slot and the resolution process continues.
The \ac{CRI} ends when all $\nuser$ users become resolved.
The \ac{BTA} algorithm can be implemented in each user by means of a simple counter~\cite{capetanakis1979tree}. 

One can represent the progression of a \ac{CRI} in terms of full-binary trees, as shown in Fig.~\ref{fig:tree_example}.
Here, we show an example with $\nuser=4$, where the users are labelled with A, B, C, D.
Each node on the tree represents a slot.
The labels inside the node show the users that have transmitted in that slot. 
A node with an empty label inside is idle, while a node with only one label inside is a success. 
The numbers outside the node represent the slot number in \ac{BTA}. 
Hence, the \ac{CRI} in this example took 9 slots. 
Leaf nodes are either successes or idle slots while all internal nodes are collisions.
The tree structure allows for a recursive analysis of the properties of the \ac{CRP}. 

\begin{figure}
    \centering
    \includegraphics[width=\columnwidth]{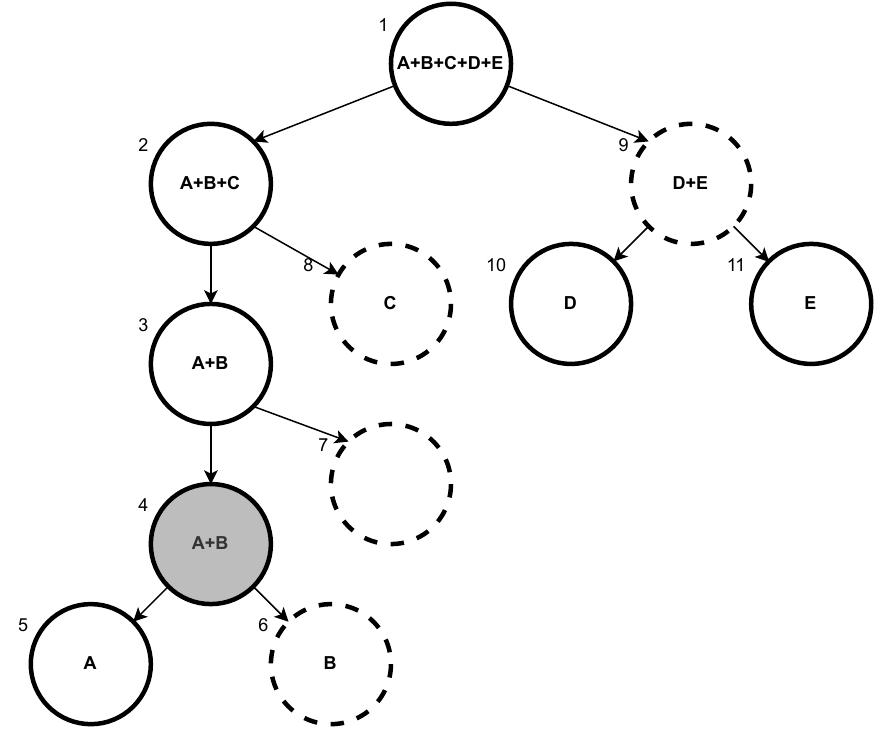}
    \caption{An example visualization of different tree algorithms. At each collision of more than 2 signals, the users independently split into two groups the users that selected the left group transmit once again. The skipped slots in SICTA are drawn with dotted lines. In \ac{ATIC}, collisions of degree two (skipped or not) are resolved deterministically. \ac{ATIC} is able to skip slot $4$, as $A$ and $B$ can infer that they are in a collision of degree two. Slot $4$ and $5$ hence become deterministic and so do $10$ and $11$.}
    \label{fig:tree_example}
\end{figure}

A full-fledged tree algorithm combines the \ac{CRP} with a \ac{CAP} that determines how and when the arriving users will transmit their packets for the first time. 
The common \ac{CAP}s are gated, windowed, and free access. 
In gated access (also called blocked access), when a \ac{CRI} is in progress, all newly arrived users wait, i.e., they are blocked. 
The blocked users transmit on the channel in the next slot 
after the current \ac{CRI} ends. 
In windowed access, the time is divided into windows whose length is $\Delta$ slots, $ \Delta \in \mathbb{R}^{+}$.
The users arriving in the $k$-th time window join the next \ac{CRI} after the \ac{CRI} of the users that arrived in the $(k-1)$-th window is over.
In free access, a newly arrived user simply transmits on the channel in the first slot after its arrival. 
It thus joins an ongoing \ac{CRI}, if there is one. 
For all the three \ac{CAP}s, the main performance parameter is the \ac{MST}. 
If the total packet arrival rate in the network $\lambda>0$ is less than the \ac{MST}, then the \ac{RA} scheme will be stable. 

It was shown in~\cite{capetanakis1979tree} that the \ac{MST} of \ac{BTA} with gated access is 0.346 packets per slot.
This \ac{MST} is improved by free access to 0.392 packets/slot~\cite{mathys1985q} and by windowed access to 0.429 packets/slot~\cite{massey1981collision}.
An improvement to the \ac{BTA} called \ac{MTA} was suggested in~\cite{massey1981collision} where a definite collision can be skipped.
If there is an idle slot after a collision, the next slot is a definite collision, and the users can avoid transmitting on this slot by randomly splitting once again. 
For example, slot 4 in Fig. \ref{fig:tree_example} is a definite collision and will be skipped in \ac{MTA}. 
Ternary \ac{MTA} with biased splitting was shown to be the optimum choice~\cite{mathys1985q} among \ac{MTA}.
In biased splitting, the probability of a user joining any of the $\tf$ groups is not uniform. 
\ac{MTA} with clipped access (a version of windowed access) was introduced in~\cite{verdu1986} and achieves \ac{MST} of 0.4878 packets/slot.

\subsection{SICTA}
The collisions in \ac{BTA} are discarded by the \ac{AP}.
One way to improve the \ac{MST} of tree algorithms was suggested in ~\cite{yu2005sicta,yu2007high}, called \ac{SICTA}. 
Here the signals of all collisions are saved by the \ac{AP}.
The \ac{AP} can then resolve more than one packet per slot by successively subtracting the signals of previously decoded packets (i.e., cancelling their interference) from the previously stored collisions.

To illustrate how certain slots from the \ac{BTA} will be skipped in \ac{SICTA}, we use the example in Fig. \ref{fig:tree_example} and use the same slot number (labeled outside the node) as the ones in the figure.
Let $Y_{t}$ be the signal (received by the AP) of slot $t$ and $X_{i}$ be the signal from user $i$.
The algorithm will proceed as \ac{MTA} until slot 6.
The \ac{AP} will save the signal from the collisions in its memory. 
Thus the 3 signals in its memory are $Y_{1} = X_{A} + X_{B} + X_{C} + X_{D}$, $Y_{2} = X_{A} + X_{B} + X_{C}$, $Y_{5} = X_{A} + X_{B}$. 
It can decode the signal from user A since it was the only user to transmit in slot 6, $Y_{6} = X_{A}$. 
As soon as there is a success in slot 6, the \ac{AP} will subtract the known signal of user $A$ from all the 3 previous saved collisions. 
After this step, the \ac{AP} can also decode $X_{B}$ since $Y_{5} - X_{A} = X_{B}$. 
The \ac{AP} then proceeds to subtract $X_{B}$ from its known collisions and is able to decode the signal of user C from the collision in the second slot. 
Finally, the \ac{AP} is also able to decode the signal from user D after subtracting $X_{C}$ from the first collision.
Since all the users from the initial collision are decoded, the \ac{CRI} is over after slot 6. 
Thus using \ac{SIC}, one can skip 4 slots out of the 9 from \ac{BTA} tree. 
The skipped slots in SICTA are marked with dotted circles.

The \ac{MST} of \ac{SICTA} with blocked access was shown to be $0.693$~\cite{yu2005sicta}.
It was also shown that neither free access nor windowed access help in increasing the \ac{MST} of \ac{SICTA} above the one of blocked access~\cite{peeters2009}.
Splitting beyond binary, i.e., $\tf-$ary \ac{SICTA} reduces the \ac{MST} for fair splitting.
However, if the probability of joining each user follows a special distribution, \ac{MST} of 0.693 can be achieved for any splitting value of $\tf$~\cite{deshpande2022correction, vogel2023analysis}.

The massive improvement of \ac{SICTA} over \ac{BTA} in terms of the \ac{MST} comes at the cost of higher complexity at the \ac{AP}.
The \ac{AP} must be able to save collision signals in its memory and be able to perform \ac{SIC} at the end of every slot.
The broadcast feedback by the \ac{AP} is also more complex, $F[t] \in \{0,k,e\}$ where $k$ is an integer indicating the number of slots that should be skipped by all active users in the \ac{CRI} after every success.

\begin{figure}[t]
    \centering
\begin{tikzpicture}

\definecolor{color0}{rgb}{0.12156862745098,0.466666666666667,0.705882352941177}

\begin{axis}[
width={\columnwidth}, height={5cm},
tick align=outside,
tick pos=left,
x grid style={white!69.0196078431373!black},
xlabel={Slot Degree},
xmajorgrids,
xmin=1, xmax=10,
xtick style={color=black},
xtick={0,1,2,3,4,5,6,7,8,9,10},
y grid style={white!69.0196078431373!black},
ymajorgrids,
ymin=0, ymax=0.578445945945946,
ytick style={color=black},
ytick={0,0.1,0.2,0.3,0.4,0.5,0.6},
yticklabels={0.0,0.1,0.2,0.3,0.4,0.5,0.6}
]
\draw[draw=none,fill=color0] (axis cs:-0.5,0) rectangle (axis cs:0.5,0);

\draw[draw=none,fill=color0] (axis cs:0.5,0) rectangle (axis cs:1.5,0);
\draw[draw=none,fill=color0] (axis cs:1.5,0) rectangle (axis cs:2.5,0.550900900900901);
\draw[draw=none,fill=color0] (axis cs:2.5,0) rectangle (axis cs:3.5,0.193243243243243);
\draw[draw=none,fill=color0] (axis cs:3.5,0) rectangle (axis cs:4.5,0.0970720720720721);
\draw[draw=none,fill=color0] (axis cs:4.5,0) rectangle (axis cs:5.5,0.0547297297297297);
\draw[draw=none,fill=color0] (axis cs:5.5,0) rectangle (axis cs:6.5,0.0382882882882883);
\draw[draw=none,fill=color0] (axis cs:6.5,0) rectangle (axis cs:7.5,0.0279279279279279);
\draw[draw=none,fill=color0] (axis cs:7.5,0) rectangle (axis cs:8.5,0.0227477477477477);
\draw[draw=none,fill=color0] (axis cs:8.5,0) rectangle (axis cs:9.5,0.0150900900900901);
\end{axis}

\end{tikzpicture}
    \caption{Simulated normalized collision degree distribution for SICTA with blocked access for 100000 collision slots. The packet arrival rate in the network is 0.693 which is the \ac{MST} for this scheme. We see that more than half the collisions have a multiplicity of 2.} 
    \label{fig:slot_degree_distr}
\end{figure}
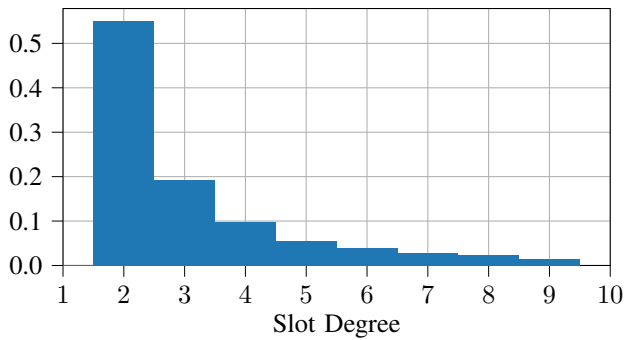

We now turn to the motivation for the scheme proposed in this paper.
Fig.~\ref{fig:slot_degree_distr} shows the distribution of the number of packets in a collision slot, denoted as the collision degree distribution, for the \ac{SICTA} algorithm with blocked access for the \ac{MST} simulated over 100000 slots.
Collisions of degree two are the most common in \ac{SICTA}.
On the other hand, the throughput\footnote{Throughput is the measure of the efficiency and is calculated as the ratio of the number of decoded packets and the number of slots required for their successful decoding.}
 of \ac{SICTA} is the lowest~\cite[Table II]{yu2007high} when resolving collision of degree 2.
In contrast, the scheme proposed in this paper has a throughput of 1 packet/slot when resolving collisions of degree 2.
This is achieved by sending more information to users in the feedback and assuming that the users are also able to perform \ac{IC},
as further elaborated in Section~\ref{sec:system_model}.



\section{System Model}
\label{sec:system_model}



We consider a large user population in the wireless range of an \ac{AP}.
Time in the system is divided into slots of equal duration.
Throughout the paper, we express all time-related quantities in terms of slots denoted by an index $t \in \mathbb{Z}^{+}$. 

Data packets arrive randomly at users, and we assume that each arrival happens at a different user and that the total number of arrivals in a slot is Poisson distributed random variable.
The users experiencing packet arrivals transmit the packets to the \ac{AP} over a shared wireless channel.
A user is said to be active when it has a packet to send to the \ac{AP}. 
Throughout this paper, we also refer to packets as users. 
All packets are of equal length. 
Each user has a unique identity, and this value is appended to the data packet. 
The transmission rate of the users is such that it takes exactly one slot to transmit one packet. 

The wireless channel is \textit{interference-limited}, modeled as the standard collision channel in the uplink as well as downlink.
Therefore, if only one user transmits in a slot, the data from the packet can be successfully decoded.
If more than one user transmits in a slot, their signals interfere, and the \ac{AP} cannot decode any of them; this scenario is called a collision. 
If no user transmits in the slot, then we say that the slot is idle.  
The \ac{AP} broadcasts immediate and instantaneous feedback to all the users in the network. 
We discuss the contents of this feedback in the next subsection. 
A user is said to be active when it has a packet in its queue to send to the \ac{AP}. 


Users are equipped with an internal memory that stores two signals: their own signal and a received signal from the \ac{AP}.

For every slot $t$, the \ac{AP} broadcasts a signal indicating the feedback $F[t]\in\{0,k,e\}$ \textit{and} a signal $Z_{t}$ to all the users in the network. If $F[t]=e$, then $Z_t=Y_t$, i.e., the signal that was just received by the \ac{AP}.
If $F[t]=k$, the \ac{AP} broadcasts $Z_t=Y_{s}-Y$, where $s<t$ and $Y_{s}$ is the most recent signal that is unresolved after performing \ac{SIC} and $Y$ is the sum of all signals resolved in between slot $s$ and slot $t$. Finally, if $F[t]=0$, then $Z_t=\emptyset$.

The users keep the last received valid $Z_{t-p} \neq \emptyset, p \in \mathbb{N}^{+} $ and their own signal $X_i$ in their memory.  
Upon a collision or success, the active users proceed to cancel their own signal $X_i$ from the broadcasted signal.
If $F[t]=e$ and the active user $X_i$ can resolve the signal  $Y_t-X_i$, $Y_t$ must have been a collision of degree two. If $F[t]=k$ and the active user $i$ can resolve the signal  $Z_t-X_i$, then $Y_s-Y-X_i$ is a decodable signal and user $i$ is in a skipped slot of degree two with the signal $Y_s-Y$.

Thus, the algorithm allows the users to identify that only two users are left in the current subtree as well as the other user's signal.
Once this information is available to both contending users, it only remains to ensure that they avoid a collision in slot $t+1$.
As mentioned previously, each user has a unique ID that is appended to the data packet.
A pre-defined hierarchical system (for example, rank users according to the value of their ID) can then be used to decide in a distributed manner which of the two contending users should transmit in slot $t+1$.
Both the users will be decoded by the \ac{AP} in slot $t+1$ as the \ac{AP} can subtract $Y_{t+1}$ from $Y_{t}$.

We illustrate \ac{ATIC} using the example from Fig.~\ref{fig:tree_example}. 
In slot 3, the \ac{AP} broadcasts $X_A+X_B$ and $F[3]=e$. This allows $A$ and $B$ to cancel their own signal from $X_A+X_B$. In the example, $A$ finds its user ID to be higher than that of $B$ and transmits in slot 5. Note that slot 4 is skipped in \ac{ATIC}, which is not the case for \ac{SICTA}.  
The same process is done by user $B$ which finds its user ID less than that of $A$'s and decides to abstain from transmitting in slot 5, fully knowing that $A$ will transmit and its own signal will be resolved along with that of user B's after slot 5. After slot $5$, the \ac{AP} broadcasts the signal $X_D+X_E$, which is the difference between $Y_1=X_A+X_B+X_C+X_D+X_E$ and the sum of the signal decoded between slot 1 and 5, which is $X_A+X_B+X_C$. Users $D$ and $E$ can then perform \ac{SIC}, with the result that $D$ finds its user ID to be higher and broadcasts first.



Note that the throughput of the algorithm described above is equivalent to the following: the \ac{AP} \textit{always} broadcasts the signal it has received, i.e. collisions and singletons. It hence does not require the \ac{AP} to calculate \textit{and} broadcast unresolved signals, for example to infer that $(X_A+X_B+X_C)-X_A$ is equal to $X_B+X_C$.
Users store all broadcasted signals in their internal memory and perform \ac{SIC} like the \ac{AP} with the additional knowledge of their own signals in the \ac{SIC} process. 
The version of \ac{ATIC} introduced previously has the same throughput as this variant. 
However, it requires users to reserve a large chunk of memory for storing the received signal and the users need have essentially the same computational capacity as the \ac{AP} to perform \ac{SIC}. 
Below we discuss a variant that does not come at the cost of increased storage capacities but uses the same feedback.
\subsection*{Simpler variant with a lower throughput}
We can still achieve a significant increase in throughput over SICTA if we apply \ac{ATIC} to only non-skipped collisions  i.e., on the left subtree. 
In this variant, the \ac{AP} only broadcasts the received signal \textit{if} it is a collision, i.e., $F[t]=e$. 
In Fig.~\ref{fig:tree_example}, this still causes slot $4$ to be skipped, as here we have a collision of degree two which is in the left group. 
However, slots 10 and 11 are no longer deterministic in this version of \ac{ATIC}, since the signal $X_D+X_E$ is actually never broadcasted. 
Therefore, we only improve throughput in collisions of degree 2, that occur on the left nodes of branches. 
Analogous to the calculations performed in Section \ref{sec:analysis}, we can show that the throughput of this variant with fair splitting is $6\log(2)/5\approx 0.832$, which is a 20\% gain over SICTA. 

\section{Analysis}
\label{sec:analysis}

Let $l_n$ be the conditional \ac{CRI} length, i.e., the \ac{CRI} length given $n$ users have collided in the first slot.
The evolution $l_n$ can be expressed recursively as,
\begin{equation}
    l_n=\begin{cases}
    1&\text{ if }n=0,1\\
    2&\text{ if }n=2\\
    l_{i}+l_{n-i}&\text{ if }n\ge 3
    \end{cases}
    \label{eq:cri_evol}
\end{equation}
where $i$ is the (random) number of users out of $n$ which have chosen group 0.
We assume that each packet chooses group 0 independently of one another and with probability $p\in(0,1)$.
We are interested in finding the expected \ac{CRI} length conditioned on the fact that $\nuser$ users have collided in the first slot, $L_n=\E\ek{l_n}$.
The expected \ac{CRI} length allows us to express the conditional throughput $T_n = \frac{\nuser}{L_n}$, which measures the average efficiency of resource utilization of the \ac{CRP}.

\subsection{Closed-form equation of $L_{\nuser}$}
\label{subsec:closed-form}


Using \eqref{eq:cri_evol} one can derive a recursive equation for $L_{n}$.
However, we skip this derivation since a closed-form expression exists for $L_n$. 

\begin{theorem}\label{Thm:ClosedExpLn}
Let $q=1-p$ and set $r={2-4pq-3\rk{p^2+q^2}}$.
We then have that for every $n\ge 0$
\begin{equation}
      L_n=1+\sum_{i= 2}^n\binom{n}{i}\frac{\rk{-1}^i\rk{i-1+ri(i-1)/2}}{1-p^i-q^i}\, \cdot
      \label{eq:thm1}
\end{equation}
\end{theorem}
\begin{proof}
    The proof is given in Appendix A.
\end{proof}

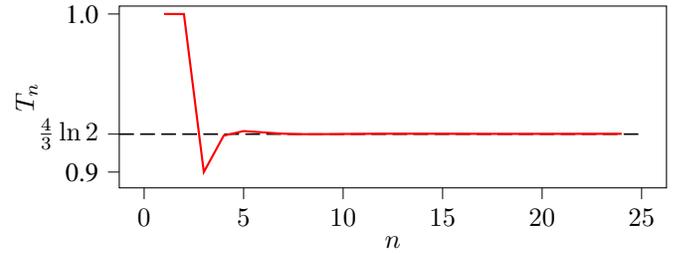
\begin{figure}[t]
    \centering
\begin{tikzpicture}

\definecolor{color0}{rgb}{0.12156862745098,0.466666666666667,0.705882352941177}

\begin{axis}[
width={\columnwidth}, height={4cm},
legend cell align={left},
legend style={fill opacity=0.8, draw opacity=1, text opacity=1, draw=white!80!black},
tick align=outside,
tick pos=left,
x grid style={white!69.0196078431373!black},
xlabel={$\nuser$},
xmin=-1.25, xmax=26.25,
xtick style={color=black},
y grid style={white!69.0196078431373!black},
ylabel={$T_{n}$},
ymin=0.89, ymax=1.005,
ytick style={color=black},
ytick={0.9, 0.924, 1.0},
yticklabels={0.9, $\frac{4}{3}\ln{2}$, 1.0},
]
\path [draw=black, semithick, dash pattern=on 5.55pt off 2.4pt]
(axis cs:-1.25,0.924)
--(axis cs:25,0.924);

\addplot [thick, red, forget plot]
table {%
1 1
2 1
3 0.9
4 0.923076923076923
5 0.925925925925926
6 0.925066312997347
7 0.924265779652766
8 0.923974676944973
9 0.923987859755998
10 0.924097597470287
11 0.924195897551348
12 0.924249260650131
13 0.924261232226415
14 0.924247458760778
15 0.924223328101189
16 0.924199576601167
17 0.92418189969132
18 0.924172115344972
19 0.924169623889469
20 0.924172637781192
21 0.924179037207664
22 0.924186877927958
23 0.924194635376452
24 0.924201273691125
};

\end{axis}

\end{tikzpicture}
    \caption{The throughput conditioned on a collision of $\nuser$ users.
    The conditional throughput is 1 when $0<\nuser\leq2$ and tends to $\frac{4}{3}\ln{2}$ as $\nuser\longrightarrow\infty$.
    The lowest conditional throughput of $0.9$ is achieved for $n=3$.} 
    \label{fig:tpt}
\end{figure}
Fig.~\ref{fig:tpt} shows the conditional throughput $T_n$ for different values of $n$ according to \eqref{eq:thm1} in the case $p=q=1/2$.
In the worst case, $T_n = 0.9$ for when $n=3$.
In comparison, in the best case, $T_n$ of \ac{SICTA} is $0.693$ or $\ln 2$.
Moreover, $T_n$ of \ac{SICTA} for $\nuser =2$ is $0.66$, while $T_n=1$ for $\nuser =2$ for the proposed algorithm.

\subsection{Asymptotic Behaviour}
\label{subsec:asymptotic_behaviour}

The closed-form expression~\eqref{eq:thm1} is useful for calculating the behavior of $L_n$ and $T_n$  when $n$ is small.
However, as $n$ grows, the number of summands in~\eqref{eq:thm1} increases and the computation becomes numerically challenging due to cancellation effects.

On the other hand, Fig.~\ref{fig:tpt} suggests that the value of $T_n$ seems to settle at $\frac{4}{3}\ln 2$ as $\nuser$ grows.
Here we formally prove that, asymptotically, $T_n$ tends to $\frac{4}{3}\ln 2$, save an oscillatory component of rather small amplitude.

\begin{theorem}\label{Thm:Asympto}
The throughput $n/L_n$ of \ac{ATIC} is asymptotically maximized for $p=q=1/2$. For this value, we have
\begin{equation}
    T_n=\frac{n}{L_n}= \frac{4}{3}\ln(2)+g(n)+o(1)\, ,
\end{equation}
where $g(n)$ is a small sine-like perturbation, as in \cite{mathys1985q}, usually between $10^{-3}$ and $10^{-6}$. 
\end{theorem}
The proof is given in Appendix B, also providing an asymptotic expansion for all values of $p$, see \eqref{Eq:NewPart}.

The asymptotics of $T_n$ can be used to derive the \ac{MST} for the case of the blocked access using the techniques from \cite{mathys1985q} or \cite{yu2007high}, where $\ac{MST} = \lim_{\nuser \rightarrow \infty} T_{n}$.
Comparing the \ac{MST} of the proposed tree algorithm with the one of \ac{SICTA} ($\log(2)$), there is a gain of one-third.


\subsection{Windowed Access}
\label{subsec:windowed_access}

In windowed access, the packets that start the $k$-th \ac{CRI} have arrived during the windowed interval $(k\Delta, (k+1)\Delta)$, where $\Delta$ is the window size 
that is optimized for the arrival rate. 
Using windowed access makes sense only if it supports a higher arrival rate than gated access. 

From Fig.~\ref{fig:tpt} we see that the conditional throughput is 1 for $0<\nuser\leq 2$.
This high efficiency for a small number of users hints to the possibility that if we can restrict most \ac{CRI}s to start with $n < 3$ by using an optimized window size, our algorithm might perform better.
Indeed, this is the case for \ac{BTA}~\cite{rom2012multiple} and tree algorithms with \ac{MPR}~\cite{ceda2020, ceda2021}, where the of the \ac{CRP} is higher for smaller $\nuser$, such that windowed access pushes the \ac{MST} to be higher than that of gated access.
For SICTA, the efficiency is in fact lower for smaller $\nuser$ and hence windowed access does not improve the \ac{MST} over blocked access~\cite{yu2007high}. 

We use the method in~\cite{rom2012multiple} to numerically find the optimal window size and the corresponding \ac{MST} for our algorithm.
The probability that $\nuser$ packets arrive in window $\Delta$ when the arrival rate is Poisson distributed with mean $\lambda$ packets per slot is then
\begin{equation}
    \Pr\{ N = \nuser \} = \frac{\left(\lambda \Delta\right)^n }{\nuser!}e^{ - \lambda \Delta}.
\end{equation}
The expected \ac{CRI} length conditioned on the window size and packet arrival rate is, 
\begin{equation}
    \ecri(\lambda\Delta) = \E \{ \ecri_\nuser | \lambda \Delta \} = \sum_{\nuser=0}^\infty   \ecri_\nuser \frac{(\lambda\Delta)^\nuser}{\nuser!} e^{-\lambda\Delta}.
\end{equation}

For the \ac{RA} scheme to be stable, the mean \ac{CRI} length has to be less than the window size $\Delta$.\footnote{$\E \{ \ecri^2_\nuser \} < \infty$ also holds, but we omit the proof which can be done following \cite{rom2012multiple}.} 
Thus, 
\begin{equation}
    L(\lambda\Delta) < \Delta.
    \label{eq:pakes_lemma}
\end{equation}
Rewriting the above equation, we get the mean arrival rate for which the \ac{RA} scheme will be stable, if $\lambda < \frac{\lambda\Delta}{L(\lambda\Delta)}$.
Fig.~\ref{fig:windowed} plots the function $\frac{\lambda\Delta}{L(\lambda\Delta)}$ for different values of $\lambda\Delta$.
The function increases to $\frac{4}{3}\ln2$ as $\lambda\Delta$ grows, implying that the optimal window size (and thus the product $\lambda \Delta$) tend to infinity.
Thus, using windowed access does not provide benefits in terms of the \ac{MST} in comparison to gated access. 

\begin{figure}
    \centering
    \input{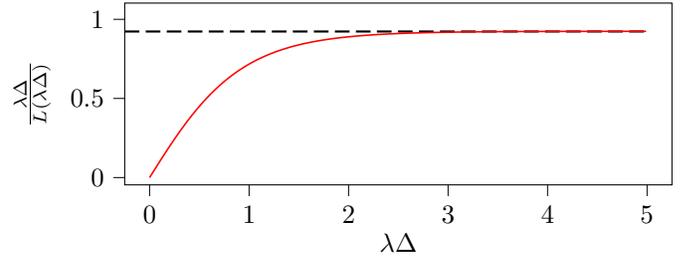}
    \caption{Throughput for windowed access with \ac{ATIC} as function of the expected number of arrivals per window $\lambda \Delta$. 
    We see that gated access is the best CAP to be used along with \ac{ATIC}. }
    \label{fig:windowed}
\end{figure}

\subsection{Delay Analysis}
\label{subsec:delay_analysis}

Similar to \cite{yu2007high,molle1992computation}, one can use the moment generating function from Equation \eqref{EquationMGF} to approximate the average delay experienced by each packet. 
As the delay does not exhibit a closed-form solution (see \cite[Section V.C]{yu2007high}), we numerically simulate the delay using the method provided in \cite{molle1992computation}.
The mean packet delay for \ac{BTA}, \ac{SICTA} and \ac{ATIC} for different mean packet arrival rates $\lambda$ is shown in Fig.~\ref{fig:Delay_figure}. 
We see that \ac{ATIC} significantly reduces the delay for arrival rates higher than 0.5 packets per slot. 
For example, at an arrival rate of $\l=0.5$, \ac{SICTA} provides a mean packet delay of 1.7 slots while \ac{ATIC} gives a mean packet delay of 1 slot only. 
This difference gets more and more pronounced as the packet arrival rate approaches the \ac{MST} of \ac{SICTA}. 
At $\l=0.693$ packets/slot, the delay of the \ac{SICTA} scheme becomes unbounded. 
For $\l$ at 95 \% of their respective \ac{MST}, \ac{SICTA} gives a mean packet delay of 12.3 slots while \ac{ATIC} gives a mean packet delay of 10.2 slots. 
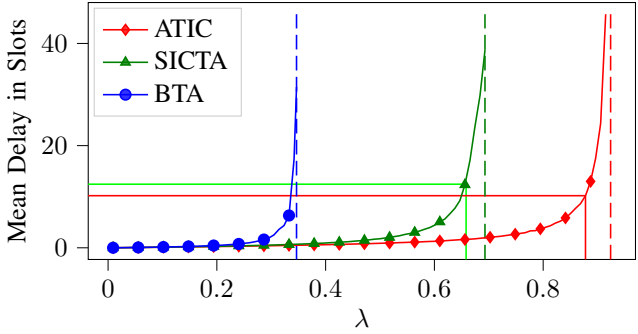
\begin{figure}
    \centering
\begin{tikzpicture}

\definecolor{color0}{rgb}{0.12156862745098,0.466666666666667,0.705882352941177}

\begin{axis}[
width={\columnwidth}, height={5cm},
tick align=outside,
legend cell align={left},
legend style={
  fill opacity=0.8,
  draw opacity=1,
  text opacity=1,
  at={(0.01,0.98)},
  anchor=north west,
  draw=white!80!black
},
tick pos=left,
x grid style={white!69.0196078431373!black},
xmin=-0.03649593947223, xmax=0.96973790187963,
xtick style={color=black},
xlabel={$\lambda$},
ylabel={Mean Delay in Slots},
ylabel style={yshift=-0.9em},
y grid style={white!69.0196078431373!black},
ymin=-2.28615048, ymax=48.00916008,
ytick style={color=black}
]
\path [draw=red, semithick, dash pattern=on 5.55pt off 2.4pt]
(axis cs:0.924,-2.28615048)
--(axis cs:0.924,45.7230096);

\path [draw=green, semithick]
(axis cs:0.6583,-2.28615048)
--(axis cs:0.6583,12.4494496);

\path [draw=red, semithick]
(axis cs:0.8778,-2.28615048)
--(axis cs:0.8778,10.2);

\path [draw=green, semithick]
(axis cs:-0.03649593947223,12.4494496)
--(axis cs:0.6583,12.4494496);

\path [draw=red, semithick]
(axis cs:-0.03649593947223,10.2)
--(axis cs:0.8778,10.2);

\path [draw=green!50.1960784313725!black, semithick, dash pattern=on 5.55pt off 2.4pt]
(axis cs:0.693,-2.28615048)
--(axis cs:0.693,45.7230096);

\path [draw=blue, semithick, dash pattern=on 5.55pt off 2.4pt]
(axis cs:0.3465,-2.28615048)
--(axis cs:0.3465,45.7230096);

\addplot [semithick, red, mark=diamond*, mark repeat={5}]
table {%
0.0092419624074 0.0092448
0.0184839248148 0.0188912
0.0277258872222 0.0286192
0.0369678496296 0.0382752
0.046209812037 0.0487088
0.0554517744444 0.0582752
0.0646937368518 0.0692816
0.0739356992592 0.0792624
0.0831776616666 0.0896448
0.092419624074 0.101984
0.1016615864814 0.112704
0.1109035488888 0.1227648
0.1201455112962 0.1359632
0.1293874737036 0.1467936
0.138629436111 0.1576128
0.1478713985184 0.1702272
0.1571133609258 0.1820448
0.1663553233332 0.1963152
0.1755972857406 0.2083632
0.184839248148 0.221136
0.1940812105554 0.2356592
0.2033231729628 0.2490784
0.2125651353702 0.2631952
0.2218070977776 0.2742224
0.231049060185 0.2914416
0.2402910225924 0.3044736
0.2495329849998 0.3186864
0.2587749474072 0.3339808
0.2680169098146 0.3500096
0.277258872222 0.3642608
0.2865008346294 0.3798256
0.2957427970368 0.3968816
0.3049847594442 0.4117712
0.3142267218516 0.430704
0.323468684259 0.4469232
0.3327106466664 0.4657648
0.3419526090738 0.485216
0.3511945714812 0.504648
0.3604365338886 0.5211024
0.369678496296 0.5413968
0.3789204587034 0.5593808
0.3881624211108 0.5808704
0.3974043835182 0.6012352
0.4066463459256 0.6225232
0.415888308333 0.6489552
0.4251302707404 0.6717184
0.4343722331478 0.6913824
0.4436141955552 0.7148992
0.4528561579626 0.7423984
0.46209812037 0.7672656
0.4713400827774 0.7916048
0.4805820451848 0.814312
0.4898240075922 0.8485856
0.4990659699996 0.878184
0.508307932407 0.906424
0.5175498948144 0.9395104
0.5267918572218 0.9677584
0.5360338196292 1.0125152
0.5452757820366 1.041576
0.554517744444 1.0822608
0.5637597068514 1.1088848
0.5730016692588 1.1486688
0.5822436316662 1.201536
0.5914855940736 1.2405664
0.600727556481 1.2692176
0.6099695188884 1.3491456
0.6192114812958 1.3850016
0.6284534437032 1.4505232
0.6376954061106 1.5185936
0.646937368518 1.5860144
0.6561793309254 1.625696
0.6654212933328 1.6865152
0.6746632557402 1.7555728
0.6839052181476 1.860536
0.693147180555 1.951592
0.7023891429624 2.06196
0.7116311053698 2.1554864
0.7208730677772 2.240128
0.7301150301846 2.3836176
0.739356992592 2.523976
0.7485989549994 2.6630448
0.7578409174068 2.8170416
0.7670828798142 2.9541536
0.7763248422216 3.31312
0.785566804629 3.352552
0.7948087670364 3.7134688
0.8040507294438 4.002408
0.8132926918512 4.322376
0.8225346542586 4.83512
0.831776616666 5.2147488
0.8410185790734 5.8485824
0.8502605414808 6.708688
0.8595025038882 7.7217056
0.8687444662956 8.6483664
0.877986428703 10.2321776
0.8872283911104 12.9962064
0.8964703535178 17.5312752
0.9057123159252 24.3512416
0.9149542783326 45.7230096
};
\addlegendentry{ATIC}
\addplot [semithick, green!50.1960784313725!black, mark=triangle*, mark repeat={5}]
table {%
0.0092419624074 0.0094
0.0184839248148 0.019128
0.0277258872222 0.0291696
0.0369678496296 0.0393392
0.046209812037 0.0499952
0.0554517744444 0.0619248
0.0646937368518 0.0725408
0.0739356992592 0.0852224
0.0831776616666 0.0967984
0.092419624074 0.110224
0.1016615864814 0.1225664
0.1109035488888 0.1347472
0.1201455112962 0.1492864
0.1293874737036 0.1654592
0.138629436111 0.177296
0.1478713985184 0.192656
0.1571133609258 0.2096096
0.1663553233332 0.22792
0.1755972857406 0.2417376
0.184839248148 0.2582976
0.1940812105554 0.276584
0.2033231729628 0.2938576
0.2125651353702 0.3160992
0.2218070977776 0.3390656
0.231049060185 0.3594352
0.2402910225924 0.3810368
0.2495329849998 0.3968368
0.2587749474072 0.4206928
0.2680169098146 0.4463056
0.277258872222 0.4695584
0.2865008346294 0.4964304
0.2957427970368 0.5251312
0.3049847594442 0.5525776
0.3142267218516 0.582488
0.323468684259 0.610456
0.3327106466664 0.6468704
0.3419526090738 0.6841488
0.3511945714812 0.7134208
0.3604365338886 0.7505952
0.369678496296 0.7967024
0.3789204587034 0.842632
0.3881624211108 0.888512
0.3974043835182 0.9228688
0.4066463459256 0.9884656
0.415888308333 1.0359392
0.4251302707404 1.0886768
0.4343722331478 1.1427984
0.4436141955552 1.20968
0.4528561579626 1.2975168
0.46209812037 1.3757696
0.4713400827774 1.457432
0.4805820451848 1.5614208
0.4898240075922 1.683096
0.4990659699996 1.7591808
0.508307932407 1.8973936
0.5175498948144 2.0427728
0.5267918572218 2.2042752
0.5360338196292 2.356952
0.5452757820366 2.6039232
0.554517744444 2.7838864
0.5637597068514 3.0178
0.5730016692588 3.2981824
0.5822436316662 3.6420224
0.5914855940736 3.9563376
0.600727556481 4.4415136
0.6099695188884 5.0763008
0.6192114812958 5.6341648
0.6284534437032 6.6536672
0.6376954061106 7.7324016
0.646937368518 9.2582688
0.6561793309254 12.3494496
0.6654212933328 17.1177168
0.6746632557402 23.79392
0.6839052181476 29.8387088
0.693147180555 38.8399344
};
\addlegendentry{SICTA}
\addplot [semithick, blue, mark=*, mark repeat={5}]
table {%
0.0092419624074 0.0094944
0.0184839248148 0.0199488
0.0277258872222 0.030728
0.0369678496296 0.0427264
0.046209812037 0.0556784
0.0554517744444 0.0691904
0.0646937368518 0.0808432
0.0739356992592 0.0968064
0.0831776616666 0.1134032
0.092419624074 0.129304
0.1016615864814 0.1469872
0.1109035488888 0.16832
0.1201455112962 0.1873168
0.1293874737036 0.2147664
0.138629436111 0.2287584
0.1478713985184 0.2659072
0.1571133609258 0.2921056
0.1663553233332 0.3196336
0.1755972857406 0.3574416
0.184839248148 0.39148
0.1940812105554 0.4429024
0.2033231729628 0.4801424
0.2125651353702 0.528072
0.2218070977776 0.6056864
0.231049060185 0.6634192
0.2402910225924 0.751208
0.2495329849998 0.875264
0.2587749474072 0.9736528
0.2680169098146 1.140512
0.277258872222 1.3532064
0.2865008346294 1.6133024
0.2957427970368 1.934056
0.3049847594442 2.5201296
0.3142267218516 3.6050208
0.323468684259 4.3525648
0.3327106466664 6.3237168
0.3419526090738 17.3315024
0.3465 31.4651312
};
\addlegendentry{BTA}
\end{axis}

\end{tikzpicture}
    \caption{The mean packet delay of BTA, SICTA, and \ac{ATIC}. The packets arrive in the network independently with a Poisson distribution with mean $\lambda$. Even for arrival rates theoretically supported by SICTA in the interval (0.6, 0.693] packets per slot, the delay becomes very large. In such cases, \ac{ATIC} provides a much lower packet delay.}
    \label{fig:Delay_figure}
\end{figure}

\section{Feedback and Memory Requirements}
\label{sec:evaluation}

\subsection{Feedback Requirements}
\label{subsec:feedback_resources}
The throughput gain in the proposed algorithm comes at the cost of requiring more resources for the downlink broadcast feedback.
Here we assess the number of resources required for the feedback for \ac{BTA}, \ac{SICTA}, and \ac{ATIC}. 
We note that this type of analysis is by default neglected in the available literature.
We also note that the required feedback resources are by default not included in the throughput calculation, which is the case in this paper as well.

For \ac{BTA}, the ternary feedback would need 2 bits/slot.
Theoretically, for \ac{SICTA}, the feedback message upon a success that contains information about how many slots (i.e., $k$) should be skipped can be any positive integer.
To estimate the number of bits one would need in practice for \ac{SICTA} in the broadcast feedback, we simulated \ac{SICTA} for the maximum supported arrival rate (at \ac{MST}). 
Fig.~\ref{fig:feedback_distr} shows the simulated probability mass distribution for $k$. 
It can be observed that the value of $k$ did not exceed 9, leading us to conclude that in practice 4 bits for the broadcast feedback in case of \ac{SICTA} can be used to represent all required feedback messages with a high probability.

In the case of \ac{ATIC}, the entire received signal must be broadcast in the feedback by the \ac{AP} in the case of a collision. 
It is reasonable to assume that the number of bits required for the feedback can be approximated by the packet size in bits $B$, and in practice, it holds that $B \gg 4$.
Moreover, in \ac{ATIC}, the feedback requirements increase with the packet size.
Nevertheless, uplink and downlink channels are of comparable capacity in a multitude of wireless cellular technologies (e.g., LTE), so this type of requirement could effectively be supported in practice.

\begin{figure}[t]
    \centering
\begin{tikzpicture}

\definecolor{color0}{rgb}{0.12156862745098,0.466666666666667,0.705882352941177}

\begin{axis}[
width={\columnwidth}, height={5cm},
tick align=outside,
tick pos=left,
x grid style={white!69.0196078431373!black},
xlabel={Feedback $k$ after a successful slot},
xmajorgrids,
xmin=0, xmax=10,
xtick style={color=black},
xtick={0,1,2,3,4,5,6,7,8,9,10},
xticklabels={0,1,2,3,4,5,6,7,8,9,10},
y grid style={white!69.0196078431373!black},
ymajorgrids,
ymin=0, ymax=0.324216765453006,
ytick style={color=black},
ytick={0,0.05,0.1,0.15,0.2,0.25,0.3,0.35},
yticklabels={0.00,0.05,0.10,0.15,0.20,0.25,0.30,0.35}
]
\draw[draw=none,fill=color0] (axis cs:-0.5,0) rectangle (axis cs:0.5,0);

\draw[draw=none,fill=color0] (axis cs:0.5,0) rectangle (axis cs:1.5,0.30002822466836);
\draw[draw=none,fill=color0] (axis cs:1.5,0) rectangle (axis cs:2.5,0.308777871860006);
\draw[draw=none,fill=color0] (axis cs:2.5,0) rectangle (axis cs:3.5,0.213660739486311);
\draw[draw=none,fill=color0] (axis cs:3.5,0) rectangle (axis cs:4.5,0.105278012983347);
\draw[draw=none,fill=color0] (axis cs:4.5,0) rectangle (axis cs:5.5,0.0471351961614451);
\draw[draw=none,fill=color0] (axis cs:5.5,0) rectangle (axis cs:6.5,0.0166525543324866);
\draw[draw=none,fill=color0] (axis cs:6.5,0) rectangle (axis cs:7.5,0.00649167372283376);
\draw[draw=none,fill=color0] (axis cs:7.5,0) rectangle (axis cs:8.5,0.00112898673440587);
\draw[draw=none,fill=color0] (axis cs:8.5,0) rectangle (axis cs:9.5,0.000846740050804403);
\end{axis}

\end{tikzpicture}
    \caption{Simulated probability mass distribution of the feedback value $k$ conditioned that the slot was a success, obtained for SICTA with blocked access for 100000 slots. The packet arrival rate $\lambda$ is 0.693, which is the \ac{MST} of the scheme. It can be observed that the value of $k$ did not exceed 9.}
    \label{fig:feedback_distr}
\end{figure}
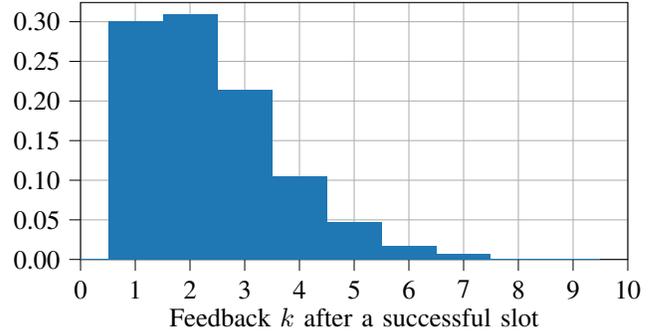

\subsection{Memory Requirements}
\label{subsec:memory_requirements}

In both SICTA and \ac{ATIC}, all collisions need to be stored in the \ac{AP}s memory. 
Following the methods given in~\cite{vogel2023analysis}, one can show that for gated access, the $C_n/L_n \sim \tfrac{1}{2}$, where $C_n$ is the expected number of collisions given $n$ packets in the initial collision. Hence, we have for ATIC 
\begin{equation}
    \frac{C_n}{n}\sim \frac{3}{8}\frac{1}{\ln(2)}\, ,
\end{equation}
which is significantly less than the $\frac{C_n}{n}\sim (2\log(2))^{-1}$ for SICTA. 
During a \ac{CRI}, the \ac{AP} needs to hold all the collisions in its memory. 
Every resolved packet is then subtracted from every saved collision to attempt \ac{IC}. 
Hence the number, of collisions in a CRI determines the memory requirements at the \ac{AP} in slots.
Thus on average, the required memory capacity of \ac{ATIC} is smaller than that of \ac{SICTA} for the same number of users starting the contention. 

Figure \ref{fig:memory_req} shows the CDF of the collisions per \ac{CRI} obtained for \ac{ATIC} and \ac{SICTA} by simulating gated access over 100000 slots. 
When the packet arrival rate, $\lambda$ is 0.693 packets/slot (\ac{MST} of \ac{SICTA}), ATIC needed a memory capacity of 7 slots.
In comparison, SICTA needed a memory capacity of 50 slots. 
On the other hand, when the arrival rate is 0.924 packets/slot (\ac{MST} of \ac{ATIC}), the required memory capacity is comparable to SICTA with $\lambda$ at 0.693 packets/slot. 
Thus, \ac{ATIC} does not require more memory capacity than \ac{SICTA} when the arrival rates are close to their respective \ac{MST}s.
When the arrival rates are the same, \ac{ATIC} needs a much smaller memory capacity. 

Note that, \ac{ATIC} also requires the users to hold 2 signals in their memory, namely their own and the received collision feedback.
In the case of \ac{SICTA} and \ac{BTA} they need to hold just 1 signal, i.e., their own. 

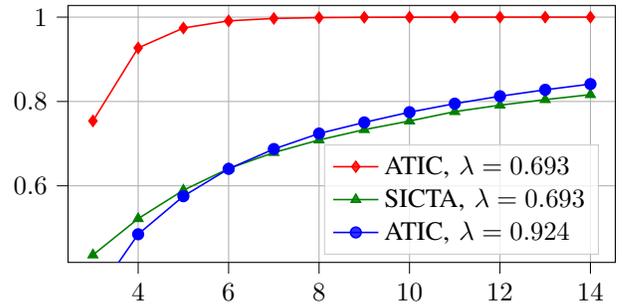
\begin{figure}
    \centering
\begin{tikzpicture}

\begin{axis}[
width={\columnwidth}, height={5cm},
legend cell align={left},
legend style={
  fill opacity=0.8,
  draw opacity=1,
  text opacity=1,
  at={(0.97,0.03)},
  anchor=south east,
  draw=white!80!black
},
tick align=outside,
tick pos=left,
x grid style={white!69.0196078431373!black},
xmajorgrids,
xmin=2.45, xmax=14.55,
xtick style={color=black},
y grid style={white!69.0196078431373!black},
ymajorgrids,
ymin=0.4184014825, ymax=1.0276951675,
ytick style={color=black}
]
\addplot [semithick, red, mark=diamond*]
table {%
3 0.7538172
4 0.9269991
5 0.97418783
6 0.99119567
7 0.99691985
8 0.99888883
9 0.99959354
10 0.99985371
11 0.99994939
12 1
13 1
14 1
};
\addlegendentry{ATIC, $\lambda = 0.693$}
\addplot [semithick, green!50.1960784313725!black,  mark=triangle*]
table {%
3 0.4360211
4 0.52213719
5 0.58961844
6 0.64094627
7 0.67889978
8 0.70870701
9 0.73321851
10 0.75367653
11 0.77590692
12 0.79124793
13 0.804391
14 0.81623966
};
\addlegendentry{SICTA, $\lambda = 0.693$}
\addplot [semithick, blue, mark=*]
table {%
3 0.330718
4 0.48515091
5 0.57540216
6 0.64022817
7 0.68724978
8 0.72401869
9 0.75037689
10 0.77463549
11 0.79486459
12 0.81243591
13 0.82771885
14 0.8412671
};
\addlegendentry{ATIC, $\lambda = 0.924$}
\end{axis}

\end{tikzpicture}
    \caption{The empirical cumulative distribution function of collisions per CRI for SICTA and ATIC obtained by simulation with gated access over 100000 slots. Since the AP has to keep every collision in a CRI in its memory, this value is indicative of the number of slots the AP needs to hold in its memory. When the arrival rate $\lambda = 0.693$, ATIC needs at the most 7 slots which is much less than the 50 needed by SICTA. When the arrival rate $\lambda= 0.924$, the required memory for ATIC is comparable to SICTA with $\lambda = 0.693$.}
    \label{fig:memory_req}
\end{figure}


\section{Conclusion}
\label{sec:conclusion}

In this paper, we have shown that by broadcasting the received composite signal and leveraging \ac{IC} at the user side, we can increase the \ac{MST} to 0.924.
This is an improvement over \ac{SICTA} by one-third and rather close to the absolute limit of 1 for the collision channel model.
Our proposed scheme also lowers the average packet delay and required memory capacity at the \ac{AP}.

The requirement of enhanced feedback limits the practical applicability of the scheme to systems with frequent and adept feedback channels and/or with short packet sizes.
An example of the latter is mobile cellular systems, such as LTE, where the capabilities and scheduling of the downlink and uplink channels are balanced by default.   



\section*{Acknowledgements}

The work of \v C. Stefanovi\' c was supported by the European Union’s Horizon 2020 research and innovation programme under Grant Agreement number 883315. 
Q. Vogel would like to thank Silke Rolles for funding part of this research. Y. Deshpande's work was supported by the
Bavarian State Ministry for Economic Affairs, Regional Development and Energy (StMWi) project KI.FABRIK under grant no. DIK0249.

\section*{Appendix}\label{Sec:Appendix}
\subsection{Proof of Theorem~\ref{Thm:ClosedExpLn}}\label{subsec:proofthm1}
We prove the statement by deriving a functional equation for the moment generating function and calculate the mean by taking the derivative.

Write for $n\ge 0$ and $Q_n(z)=\E\ek{z^{l_n}}$, where $z\in \C$. Equation \eqref{eq:cri_evol} gives
\begin{equation}
    Q_0(z)=Q_1(z)=z\quad\text{and}\quad Q_2(z)=z^2\, .
\end{equation}
Note that $i$, the number of packets choosing the left slot, follows a binomial distribution with parameter $p$. By conditioning on $i$ and using Equation \eqref{eq:cri_evol}, we get for $n\ge 3$
\begin{equation}
    Q_n(z)=\E\ek{z^{l_n}}=\sum_{i=0}^n\binom{n}{i}p^i\rk{1-p}^{n-i}Q_i(z)Q_{n-i}(z)\, .
\end{equation}
Write now $q=1-p$, for brevity. Using the two equations above, we get for the Poisson moment generating function $Q(x,z)=\sum_{n\ge 0}x^n Q_n(z)/n!$ that
\begin{multline}\label{EquationMGF}
    Q(x,z)=(1+x)z+\frac{x^2z^2}{2}\\
    +\sum_{n\ge 3}\sum_{i=0}^n\rk{\frac{Q_i(z)(px)^i}{i!}}\rk{\frac{Q_{n-i}(z)(xq)^{n-i}}{(n-i)!}}\, .
\end{multline}
Note that the final term can be rewritten as
\begin{multline}
    \sum_{n\ge 0}\sum_{i=0}^n\rk{\frac{Q_i(z)(px)^i}{i!}}\rk{\frac{Q_{n-i}(z)(xq)^{n-i}}{(n-i)!}}\\-z^2-z^2x-x^2\rk{z^3p^2/2+z^2pq+z^3q^2/2}\, ,
\end{multline}
by adding and subtracting the terms for $n=0,1,2$. The previous two equations give
\begin{multline}\label{EquationFunRelQ}
  Q(x,z)=Q(px,z)Q(qx,z)+\rk{z-z^2}\\+x\rk{z-z^2}+\frac{x^2}{2}\rk{z^2-2z^2pq-z^3\rk{p^2+q^2}}\, .  
\end{multline}
The moment generating function of $L_n$ is given by $L(x)=\ex^{-x}\sum_{n\ge 0}x^n L_n/n!$. By differentiation, one obtains that $L(x)=\ex^{-x}\frac{\d Q}{\d z}(x,1)$. Hence, Equation \eqref{EquationFunRelQ} gives
\begin{multline}
    L(x)=L(px)+L(qx)\\-\ex^{-x}\rk{1+x-\frac{x^2}{2}\rk{2-4pq-3\rk{p^2+q^2}}}\, .  
\end{multline}
Set $r=\rk{2-4pq-3\rk{p^2+q^2}}$. We then have that for $L(x)=\sum_{n\ge 0}\alpha_n x^n$ and $n\ge 3$
\begin{equation}
    \alpha_n=\frac{1}{n!}\frac{\rk{-1}^n\rk{n-1+rn(n-1)/2}}{1-p^n-q^n}\, ,
\end{equation}
by coefficient comparison. Note that by the definition of $l_n$,$\alpha_0=1$, $\alpha_1=0$ and $ \alpha_2=1/2$. Using that $L_n=\sum_{i=0}^n \alpha_i n!/(n-i)!$ immediately gives
\begin{equation}
    L_n=1+\frac{n(n-1)}{2}+\sum_{i= 3}^n\binom{n}{i}\frac{\rk{-1}^i\rk{i-1+ri(i-1)/2}}{1-p^i-q^i}\, .
\end{equation}
We rewrite the right-hand side as
\begin{multline}
    1+\frac{n(n-1)}{2}\rk{1-\frac{1+r}{1-p^2-q^2}}+\\ \sum_{i= 2}^n\binom{n}{i}\frac{\rk{-1}^i\rk{i-1+ri(i-1)/2}}{1-p^i-q^i}\, .
\end{multline}
Note that $1-\frac{1+r}{1-p^2-q^2}=2\frac{-1+(p+q)^2}{1-p^2-q^2}=0$. Hence
\begin{equation}\label{EquationClosedFormAppendix}
      L_n=1+\sum_{i= 2}^n\binom{n}{i}\frac{\rk{-1}^i\rk{i-1+ri(i-1)/2}}{1-p^i-q^i}\, .
\end{equation}
 This concludes the proof of Theorem \ref{Thm:ClosedExpLn}.\qed
\subsection{Proof of Theorem~\ref{Thm:Asympto}}\label{subsec:proofthm2}
The starting point of our asymptotic analysis of $L_n$ is Equation \eqref{EquationClosedFormAppendix}. Note that by \cite[Equation 34]{yu2007high}
\begin{equation}\label{Eq:OldPart}
    \sum_{i= 2}^n\binom{n}{i}\frac{\rk{-1}^i\rk{i-1}}{1-p^i-q^i}\sim n\frac{1}{-p\log(p)-q\log(q)}+ng_1(n)\, ,
\end{equation}
where $g_1(n)$ is a small fluctuation as described in Theorem~\ref{Thm:Asympto}.

We now analyse the remaining part of the sum:
\begin{equation}
    \frac{r}{2}\sum_{i= 2}^n\binom{n}{i}\frac{\rk{-1}^i{i(i-1)}}{1-p^i-q^i}\, .
\end{equation}
As $r/2$ is a linear factor, we neglect it for now and multiply it back on later.

By differentiation for the binomial theorem, one obtains
\begin{equation}
    \sum_{i=2}^n\binom{n}{i}i(i-1)x^i=x^2n(n-1)\ek{1+x}^{n-2}\, .
\end{equation}
We write $\Pfrak_m\hk{2}=\gk{\mu_1,\mu_2\in\N\cup\gk{0}\colon \mu_1+\mu_2=m}$. We also abbreviate $p(\mu)=p^{\mu_1}q^{\mu_2}$. Using the geometric series, we get that
\begin{multline}
    \sum_{i= 2}^n\binom{n}{i}\frac{\rk{-1}^i\rk{i(i-1)}}{1-p^i-q^i}\\
    =\sum_{m\ge 0}\sum_{\mu\in\Pfrak_m\hk{2}}\binom{m}{\mu}p(\mu)^2n(n-1)\rk{1-p(\mu)}^{n-2}\, .
\end{multline}
Using a similar reasoning to \cite{mathys1985q}, we extract the leading term
\begin{equation}
    \sum_{m\ge 0}\sum_{\mu\in\Pfrak_m\hk{2}}\binom{m}{\mu}p(\mu)^2n^2\rk{1-p(\mu)}^{n}\, .
\end{equation}
Using the same reference again, we write this as
\begin{equation}\label{eq:sumBeforeMellin}
    (1+o(1))\sum_{m\ge 0}\sum_{\mu\in\Pfrak_m\hk{2}}\binom{m}{\mu}p(\mu)^2n^2\ex^{-n p(\mu)}\, .
\end{equation}
Recall that for a function $f$, its Mellin transform (see \cite{FLAJOLET19953}) is given by
\begin{equation}
    \Mcal\ek{f(x);s}=\int_0^\infty x^{s-1}f(x)\d x
\end{equation}
with inverse transform
\begin{equation}
    f(x)=\frac{1}{2\pi i}\int_{c-i\infty}^{c+i\infty}x^{-s}\Mcal\ek{f(x);s}\d s\, ,
\end{equation}
for some $c\in\R$. For $f(x)=x^2\ex^{-x}$, we have that it is 
\begin{equation}
    \Mcal\ek{f(x);s}=\Gamma(s+1)\quad\text{if}\quad\Re(s)>-2\, .
\end{equation}
Hence, we get that the sum in Equation \eqref{eq:sumBeforeMellin} can be rewritten as
\begin{equation}
     \sum_{m\ge 0}\sum_{\mu\in\Pfrak_m\hk{2}}\binom{m}{\mu}\frac{1}{2\pi i}\int_{-3/2-i\infty}^{3/2+i\infty}n^{-s}p(\mu)^{-s}\Gamma(s+1)\d s\, .
\end{equation}
We use the geometric sum again to rewrite the above as
\begin{equation}
    \frac{1}{2\pi i}\int_{-3/2-i\infty}^{3/2+i\infty}\frac{n^{-s}\Gamma(s+1)}{1-p^{-s}-q^{-s}}\d s\, .
\end{equation}
Using the reside theorem again as in \cite{mathys1985q}, we get the the above integral is given by
\begin{equation}\label{Eq:NewPart}
    n\rk{\frac{1}{-p\log(p)-q\log(q)}+g_2(n)}\, ,
\end{equation}
where $g_2(n)$ is again small and fluctuating. Combining Equation \eqref{Eq:OldPart} and Equation \eqref{Eq:NewPart}, we get that
\begin{equation}\label{Eq:FinalResult}
    \frac{L_n}{n}=\frac{1+r/2}{-p\log(p)-q\log(q)}+g_3(n)\, .
\end{equation}
In the case of fair splitting, we get that $r=-1/2$ and hence 
\begin{equation}
    \frac{n}{L_n}\quad\text{has leading term}\quad\frac{4}{3}\log(2)\approx 0.9242\, .
\end{equation}
Using \cite{fayolle1986functional}, we note that $g_3(n)$ is zero if and only if $\log\rk{ p/q}$ is irrational. Note that by expanding $r$, the leading term in Equation \eqref{Eq:FinalResult} can be written as
\begin{equation}
    \frac{p^2 - p - 0.5}{(1 - p)\log(1 - p) + p\log(p)}\, ,
\end{equation}
and is hence minimized at $p=1/2$. This concludes the proof of Theorem \ref{Thm:Asympto}.\qed
\bibliographystyle{IEEEtran}
\bibliography{Bibliography}

\end{document}